\documentclass[a4paper,11pt]{article}
\usepackage[utf8]{inputenc}
\usepackage[english]{babel}
\usepackage[top=3cm, bottom=3.5cm, right=3cm, left=3cm]{geometry}

\usepackage{amsthm}
\newtheorem{theorem}{Theorem}
\newtheorem{prop}{Proposition}
\newtheorem{lem}{Lemma}

\theoremstyle{definition}
\newtheorem{definition}{Definition}
\theoremstyle{remark}
\newtheorem*{remark}{Remark}
\usepackage{graphicx}

\usepackage{amssymb, amsmath, amsfonts}
\usepackage{url}
\usepackage{wasysym}
\usepackage{tikz}
\usetikzlibrary{positioning} 
\usetikzlibrary{shapes} 
\usetikzlibrary{arrows,decorations.markings} 
\tikzstyle directed=[postaction={decorate,decoration={markings,mark=at position .65 with {\arrow[scale=1.2]{>}}}}]

\usepackage{color} 
\usepackage{ulem} 
\newcommand{\xor}{\oplus}
\newcommand{\Xor}{\bigoplus}
\newcommand{\erase}[1]{}

\newcommand{\B}{\mathbb{B}}
\newcommand{\C}{\mathcal{C}}
\newcommand{\F}{\mathcal{F}}

\newcommand{\N}{\mathcal{N}}

\newcommand{\1}{\mathbf{1}}

\newcommand{\ie}{\textit{i.e. }}
\newcommand{\eq}{\Leftrightarrow}
\renewcommand{\neg}[1]{\overline{#1}}
\renewcommand{\leq}{\leqslant}
\renewcommand{\geq}{\geqslant}

\begin{document}
  \title{On the flora of asynchronous locally non-monotonic Boolean automata networks} 
  \author{
  Aurore Alcolei \thanks{\'Ecole Normale Sup\'erieure, Lyon, France; 
  \texttt{aurore.alcolei@ens-lyon.fr}} 
  \and K\'evin Perrot 
  \thanks{Aix-Marseille Universit\'e, CNRS, LIF UMR7279, Marseille, 
France;  \texttt{kevin.perrot@lif.univ-mrs.fr}} 
  \and Sylvain Sen\'e 
  \thanks{Aix-Marseille Universit\'e, CNRS, LIF UMR7279, Marseille, 
France; \texttt{sylvain.sene@lif.univ-mrs.fr}}
}
\date{}

\maketitle

\begin{abstract} 
Boolean automata  networks (BANs)  are   a  well established model  for 
regulation systems such as neural networks or  gene regulation  networks. 
Studies on the  asynchronous dynamics of BANs have mainly focused on 
monotonic networks, where fundamental questions on the links relating their 
static and dynamical properties have been raised and addressed. 
This paper explores analogous questions on non-monotonic networks, $\xor$-BANs 
(xor-BANs), that are BANs where all the local transition functions are 
$\xor$-functions.
Using algorithmic tools, we give a general characterisation of the asynchronous 
transition graphs of most of the strongly connected $\xor$-BANs and cactus 
$\xor$-BANs. As an illustration of these results, we provide a complete 
description of the asynchronous dynamics of two particular structures of 
$\xor$-BANs, namely $\xor$-Flowers and $\xor$-Cycle Chains. 
This work also draws new behavioural equivalences between BANs, using rewriting 
rules on their graph description.
\end{abstract}


\section{Introduction}\label{intro}

Boolean automata networks (BANs) are discrete interaction networks that are now 
well established models for biological regulation systems such as neural 
networks~\cite{h82,h84} or gene regulation networks~\cite{k69,t73}.
To this extent, locally monotonic BANs have been widely studied, both on the 
applied side~\cite{g08,mta99} and on the theoretical 
side~\cite{jlv10,mrs13,noual_th,nrs12b,rmct03}. 
However, recent works have brought new interests in local non 
monotony~\cite{nrs13}. 

On the biological side, it has been shown that, sometimes, gene regulations 
imply more complex behaviour than what is usually assumed: this is for example 
the case when one also takes in account the effect of their 
byproducts~\cite{tt95}. 
In this case, local non monotony may be required for modelling, in particular 
because this allows to express sensitivity to the environment. 

On the theoretical side, it has been noticed~\cite{n11,nrs12b} that non local 
monotony is often involved when it comes to singular behaviours in BANs. 
For example it has been shown that the smallest network that is not robust to 
the addition of synchronism (\ie allowing some automata to update 
simultaneously) is a locally non-monotonic BAN~\cite{n11,nrs12b}.

In the lines of~\cite{nrs13}, the present study is a first step towards a 
better understanding of locally non-monotonic BANs. 
It focuses on $\xor$-BANs, that is, BANs in which the state of an automaton $i$ 
is updated by xoring the state value (or the negated state value) of the 
incoming neighbours of $i$. In other words, in these BANs, every local 
transition function is of the form $f_i = \Xor_{j\in N^+(i)} \sigma(x_j)$ where 
$\sigma\in\{id, neg\}$ and $N^+(i)$ denotes the set of incoming neighbours of 
$i$~\cite{nrs12b}. 

Following a constructive approach, we first looked at some particular BAN 
structures that combine cycles, such as the double-cycle 
graphs~\cite{dns12,mrrs15}, the flower-graphs~\cite{dr12} and the cycle chains. 
All these BANs belong to the family of cactus BANs since any two simple 
cycles in their structure have at most one automaton in common.
Actually, we realised that most of the specific results we got for each 
of these BANs could in fact be generalised to a wide set of $\xor$-BANs: the 
strongly connected $\xor$-BANs with an induced double cycle of size greater than 
$3$.

A precise specification of these BANs is given in Section~\ref{section_def}. 
This section also introduces all the definitions and notations that will be used 
in the sequel. 
Section~\ref{section_res} is dedicated to the presentation and proofs of the 
general results obtained about the asynchronous dynamics of strongly connected 
$\xor$-BANs with an induced double cycle of size greater than $3$. Similarly to 
what is done in~\cite{mrrs15}, these results are based on an algorithmic 
description of the asynchronous transition graph of these BANs. 
We conclude this paper in Section~\ref{section_ex} with a full characterisation 
of two types of $\xor$-BANs, the $\xor$-flower BANs and the $\xor$-cycle chain 
BANs, which illustrates the results of Section~\ref{section_res} and provides 
new behavioural equivalences bisimulation results specific to $\xor$-BANs. 
These last results are of interest since they provide new perspectives for BAN 
classification through the use of rewrites of their interaction graphs.

\erase{
}

\section{Definitions and notations} \label{section_def}

\subsection{Static definition of a BAN}

A BAN is defined as a set of Boolean automata that interact with each other. 
The \emph{size} of a network corresponds to the number of automata in it. For a 
network $\N$ of size $n$ we denote  $V = \{1,\ldots,n\}$ the corresponding set 
of automata. 

A \emph{Boolean automaton} $i$ is an automaton whose state has a Boolean value 
$x_i \in \B = \{0,1\}$.
The Boolean vector $x=(x_i)_{i=1}^n$ that gathers together the states of all 
automata in the  network is called a \emph{configuration} of $\N$.
In the following, we will sometimes denote by $x[i,k]$ the subvector that 
records the states of the automata from $i$ to $k$, for $i<k$.
We will shorten by $\neg{x}^i$ the configuration $x$ where the state of the 
$i^{th}$ automaton is negated, and similarly, for any subset $I$ of $V$, 
$\neg{x}^I$ will denote the configuration $x$ where the states of the automata 
in $I$ are negated.

The state of an automaton can be \emph{updated} according to its \emph{local 
transition function} $f_i: \B^n \to \B$. This local function characterises how 
the automaton reacts in a given configuration: just after being updated, the 
state of $i$ has value $f_i(x)$ where $x$ is the configuration of the network 
before the update.
We say that $i$ is \emph{stable} in $x$ if $f_i(x) = x_i$. It is \emph{unstable} 
otherwise. 
Hence a network $\N$ is completely described by its set of local transition 
functions $\N = \{f_i\}_{i=1}^n$.

An automaton $i$ is said to be an \emph{influencer} of an automaton $j$ if there 
exists a configuration $x$ such that $f_j(x) \neq f_j(\neg{x}^i)$. In this case $j$ 
is said to be \emph{influenced by} $i$. We denote by $I_j$ the set of influencers of 
$j$.

In a BAN, a \emph{path} $\pi = i_0 i_1\ldots i_k$ of \emph{length} $k$ is a sequence 
of distinct automata such that for all $1\leq j\leq k$, $i_{j-1}\in I_j$. 
A BAN is \emph{strongly connected} if there is a path between every two 
automata. 
A \emph{nude path} is a particular path such that for all $1\leq j\leq k$, 
$i_{j-1}$ is the unique influencer of $i_j$ ($I_j=\{i_{j-1}\}$), \ie $f_j(x) = 
x_{j-1}$ or $f_j(x) =\neg{x_{j-1}}$. We define the \emph{sign} of a nude path as 
the parity of the number of local functions of the form $f_i(x) = \neg{x_{i-1}}$ 
that compose it,\ie $sign(\pi) = \left( \sum_{j=1}^n \1_{f_j(x) = 
\neg{x_{j-1}}}\right)\mod(2)$. 
A nude path is \emph{maximal} if any extension of it is not a nude path. 
We will denote by $\pi_i$ the maximal nude path that ends in automaton $i$. Paths and 
nude paths get their name from the graphical representation that is often associated 
to BAN as we will see next.

To get a sense of what a network looks like, it is common to give a graphical 
representation of it. To every local functions $f_i$, one can associate a Boolean 
formula $\F_i$ over the variables $x_i$. The literal associated to the $k^{th}$ 
occurrence of the variable $x_i$ is denoted by $\sigma_k(x_i)$ where $\sigma_k$ is 
the sign of the literal. 
Then the \emph{interaction graph} of $\N$ according to these formulas is the 
signed directed graph $G = (V,A)$, where $V = \{1,\ldots,n\}$ is the set of 
nodes of $G$ with one entry points per literal in $\F_i$, and $A$ is the set of 
arcs defined by $(i,j,\sigma_k) \in A$ if the $k^{th}$ occurrence of the 
variable $x_i$ in $\F_j$ has sign $\sigma_k$ (see Figure~\ref{schema_ban} 
\emph{(a)}).

As we focus on $\xor$-BANs, all formula $\F_i$ involving more than one 
automaton 
will be written in Reed-Muller canonical form, that is $\F_i = \Xor_{j\in I_i} 
\sigma_j(x_j)$. 
The {\em type} of a BAN will refer to the underlying structure of its 
interaction graph (modulo the sign of the literals and a renaming of the 
automata). A type of BANs can be described by a family of graphs, and we will 
say that two BANs are of the same {\em type} if their interaction graphs 
are isomorphic (we ignore the labels).

The simplest interaction structure that allows for complex behaviour is the 
cycle 
structure~\cite{r86}. A \emph{Boolean automata cycle} (BAC) $\C$ of size $n$ is 
a BAN 
defined as a set of local functions $\{f_i\}_{i=1}^{n}$ such that $f_i(x) =  
x_{((i-1)\mod(n))} $ or $f_i(x) = \neg{x_{((i-1)\mod(n))}}$ for all 
$i\in\{1,\ldots,n\}$. Abusing notation we will often express $f_i$ via its 
formula 
representation $\F_i = \sigma_i(x_{pred(i)})$ where $pred(i) = (i-1 \mod(n))$ 
is the only influencer of $i$ in $\C$ and $\sigma_i$ is its sign (either the 
identity or the negation function).

In the following, the majority of the networks or patterns we discuss are made 
of 
cycles that intersect each other. If an automaton $i$ is the intersection of 
$\ell$ 
distinct cycles, then its local transition function will be $f_i(x) = 
\Xor_{j=1}^\ell 
\sigma_j(pred_{j}(i))$ where $pred_{j}(i)$ represents the predecessor of $i$ in 
each 
of the incident cycles.

If a BAN is described in terms of intersections of $m$ simple cycles, $\C_1,
\ldots,\C_m$, we will often represent its size by a vector of natural numbers $n 
= 
(n_1,\ldots,n_m)$, where $n_k$ is the size of the $k^{th}$ cycle. We will also 
use 
this vector representation to describe the configurations of the BAN: $x = 
(x^1, 
\ldots, x^m) \in \B^{n_1}\times\ldots\times\B^{n_m}$ will represent the 
configuration 
where each cycle $\C_k$ is in configuration $x^k\in \B^{n_k}$. By extension 
$x^k_j$ 
will denote the state of automaton $i^k_j$ which is the $j^{th}$ automaton of 
cycle 
$\C_k$. 

As one can expect, a {\em strongly connected} $\xor$-BAN is a $\xor$-BAN whose
interaction graph is strongly connected. Hence the type of these BANs can always 
be 
described as a set of simple cycles and intersection automata. 
Strongly connected \emph{cactus} BANs are special strongly connected BANs where 
any two simple cycles intersect each other at most once~\cite{emc88}. The 
simplest example of BANs of this form are the \emph{$\xor$-Boolean automata 
double-cycles} ($\xor$-BADCs). These $\xor$-BANs are described by two cycles 
$\C_1$, $\C_2$ that intersect at a unique automaton $o = i^1_1 = i^2_1$. The 
$\xor$-BAN depicted in Figure~\ref{schema_ban} \emph{(a)} is in fact a 
$\xor$-BADC of size $(2,1) = 2 + 1 - 1 = 2$. 

\subsection{Asynchronous dynamics of a BAN}

As previously mentioned, the configuration of a network may change in time along 
with 
the local updates that are happening. A local update is formally described by a 
subset $W$ of $V$ which contains the automata to be updated at a time. We say 
that 
$W$ is \emph{asynchronous} if it has cardinality 1, that is, $W = \{i\}$ for 
some $i 
\in V$.

An update $W$ makes the system move from a configuration $x$ to a configuration 
$x'$ 
where $x'_i = f_i(x)$ if $i\in W$, and $x'_i = x_i$  otherwise. This defines a 
global 
function $F_W:\B^n\to\B^n$ over the set of configurations.

A network evolves according to a particular \emph{mode} $M\subseteq 
\mathcal{P}(V)$ 
if all its moves are due to updates from $M$. The \emph{asynchronous mode} of a 
BAN 
of size $n$ is then defined by the set $A = \{\{i\}\}_{i=1}^n$ of 
asynchronous updates, it is non-deterministic. Note that our definition of 
update mode is not fully general \cite{noual_th} but sufficient for the scope 
of this paper.

We say that a configuration $x'$ is \emph{reachable from} a configuration $x$ 
(in a 
mode $M$) if there exists a finite sequence of updates $(W_t)_{t=1}^s$ (in $M$) 
such 
that $F_{W_1}\circ\ldots\circ F_{W_s}(x) = x'$. Then, a configuration is 
\emph{unreachable} (in $M$) if it cannot be reached from any other configuration 
but 
itself (in $M$). Finally a \emph{fixed point} (of $M$) is a configuration $x$ 
such 
that $F_W(x) = x$ for every update $W$ (in $M$).

The study of the dynamics of a network under a particular update mode aims at 
making 
predictions, \ie given an initial configuration $x$, we want to tell what are 
the 
possible sets of configurations in which the network can end asymptotically. 
These 
sets are called \emph{attractors} of the network and the set of configurations 
from 
which they can be reached are their \emph{attraction basins}. Notice that a 
fixed 
point is an attractor of size 1.

The dynamics of a network $\N$ according to an update mode $M$ can be modelled 
by a 
labelled directed graph $G^M_\N = (\B^n,\bigcup_{W\in M} F_W)$, called the 
\emph{M-transition graph} of $\N$, such that:
\begin{itemize}
\item the set of vertices $\B^n$ corresponds to the $2^n$ configurations of 
$\N$.
\item the arcs are defined by the transition graph of the functions $F_W$ for 
all 
	$W\in M$, that is, $x \overset{W}{\longrightarrow} x'$ is an arc of $G$ if 
and only 
	if $W\in M$ and $F_W(x) = x'$.
\end{itemize}
The transition graph $G^A_\N$ associated to the asynchronous update mode is 
called 
the \emph{asynchronous transition graph} of $G$, shorten ATG. 
Figure~\ref{schema_ban} 
\emph{(b)} shows the ATG of the $\xor$-BADC depicted on the left.

\begin{figure}[t!]
	\centerline{
		\scalebox{1}{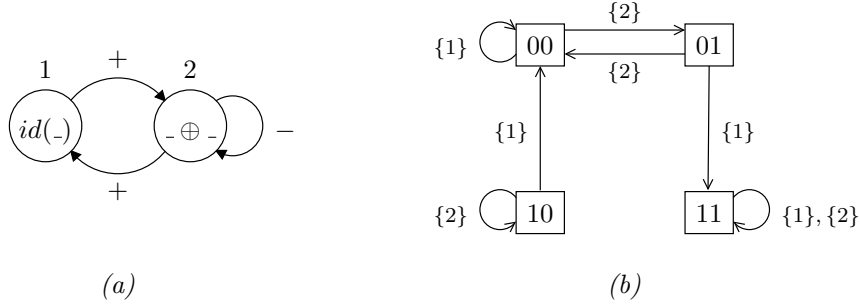}
	}
	\caption{\emph{(a)} The interaction graph of BAN $\{f_1(x) = x_2, f_2(x) = 	
x_1 
		\xor \neg{x_2}\}$ and \emph{(b)} its asynchronous transition graph.}
	\label{schema_ban}
\end{figure}

In terms of transition graphs, an \emph{attractor} of $\N$ for the mode $M$ 
corresponds to a \emph{terminal} strongly connected component of $G^M_\N$, that 
is, a 
strongly connected component that does not admit any outgoing arcs. The 
attraction 
basin of an attractor corresponds to the set of configurations in $G^M_\N$ that 
are 
connected to this component. 

In a mode $M$, the configurations that do not pertain to an attractor are 
called 
\emph{transient} configurations. These configurations can be \emph{reversible} 
or \emph{irreversible} depending on whether it is possible to reach them again 
once they have been passed.
A particular type of irreversible configurations are the \emph{unreachable} 
configurations that are the configurations that do not have any incoming arcs 
but self-loops in $G^M_\N$.

Because of the correspondence between transition graphs and dynamics, most of 
the results presented in the following are expressed in terms of walks and 
descriptions of the asynchronous transition graphs of the networks we study.

\subsection{Behavioural isomorphism Bisimulation equivalence relation}

We conclude this section with a quick reminder on \emph{behavioural 
isomorphism} which is an equivalence relation over the set of BANs that 
expresses the fact that two networks ``behaves the same way'' (up to a renaming 
of their automata and/or of their configurations). More precisely, the 
equivalence of $\N$ and $\N'$ means that, for any update mode $M$, the 
transition graphs $G^M_\N$ and $G^M_{\N'}$ are isomorphic.

\begin{definition}\label{def_bisim}
	Two BANs $\N$ and $\N'$ are (behaviourally) isomorphic if there exist two 
bijections 	$\varphi: V \to V'$ over the set of automata and $\phi: \B^n \to 
\B^n$ over the set 	of configurations such that for any update $W\subseteq V$ 
in 
$\N$, the 	corresponding update $\varphi(W)$ acts the same way in $\N'$, that 
is, for all configurations $x$, $\phi(F_W(x)) = F'_{\varphi(W)}(\phi(x))$.
\end{definition}

This definition of isomorphism between BANs has been first
introduced in~\cite{noual_th} under the name of bisimulation. 
We recall here some general results about it.

\begin{theorem}[\cite{noual_th}]\label{thm_bisim_neg}
	Let $\N=\{f_i\}_{i=1}^n$ be a BAN and $\N^\bot = \{f^\bot_i\}_{i=1}^n$ be 
its 
	\emph{dual} network defined as $f^\bot_i(x) = \neg{f_i(x)}$ then $\N$ and 
	$\N^\bot$ are isomorphic. 
\end{theorem}

\begin{theorem}[\cite{noual_th}]\label{thm_bisim_can}
	Let $\N=\{f_i\}_{i=1}^n$ be a BAN and $\N^+ = \{f^+_i\}_{i=1}^n$ be its 
	\emph{canonical} network defined as \emph{(i)} $f^+_i(x) = x_j$ if $f_i(x) = 
x_j$ or 
	$\neg{x_j}$, and \emph{(ii)} $f^+_i(x) = f_i(\neg{x}^I)$ otherwise, where 
$I 
	= \{i\in V\ |\ sign(\pi_i) = 1\}$ is the set of automata whose maximal 
incoming 
	nude path has negative sign. Then $\N$ and $\N^+$ are isomorphic.
\end{theorem}

Theorem~\ref{thm_bisim_neg} is of importance because it tells us that all the 
results 
stated in the sequel will also hold for $\eq$-BANs, which are the dual BANs of 
the 
$\xor$-BANs since all their local functions are of the form $f_i(x) = 
\underset{j\in I_i}{\eq} \sigma(x_j)$. 
On the other side, Theorem~\ref{thm_bisim_can} is very useful when studying 
particular types of networks because it greatly reduces the number of cases to 
study. Indeed, it says that one only needs to focus on networks with positive 
nude paths to characterise the whole set of possible transition graphs for a 
given type of networks. For example, it states that there are only three 
different cases of $\xor$-BADCs to study: the \emph{positive} ones, the 
\emph{negative} ones and the \emph{mixed} ones, that respectively correspond to 
the 
case where $f_o(x) = x_1^1 \xor x_1^2$,  $f_o(x) = \neg{x_1^1} \xor \neg{x_1^2}$ 
and  
$f_o(x) = \neg{x_1^1} \xor x_1^2$. 
There is actually only one class of $\xor$-BADCs since: \emph{(i)} the equality 
$x_1^1 \xor x_1^2 = \neg{x_1^1} \xor \neg{x_1^2}$ implies that positive and 
negative $\xor$-BADCs are trivially isomorphic; \emph{(ii)} a 
positive $\xor$-BADC  is isomorphic to a mixed $\xor$-BADC of same structure by 
taking $\phi(x) = \neg{x}^V$.

To prove that two networks are isomorphic we will often use a stronger 
condition 
than the one given in Definition~\ref{def_bisim}.

\begin{lem} \label{lem_bisim_cond}
	Two BANs $\N = \{f_i\}_{i=1}^n$, $\N' = \{f'_i\}_{i=1}^n$ are 
isomorphic if and only if there exists a bijection 
$\varphi: \{1,\ldots,n\} \to \{1,\ldots,n\}$ and a set 	$\{\phi_i: \B \to 
\B\}_{i=1}^n$ of (non constant) Boolean functions such that for all automata 
$i$, $\phi_i \in \{id, neg\}$, and for all configurations $x\in \B^n$, 
$\phi_i(f_i(x)) = f'_{\varphi(i)}(\phi(x))$ where $\phi(x)$ is defined 
componentwise by $\phi(x)_i = \phi_{\varphi^{-1}(i)}(x_{\varphi^{-1}(i)})$.
\end{lem}

\begin{proof}
  The proof of the right implication is straightforward since the equality  
$\phi_i(f_i(x)) = f'_{\varphi(i)}(\phi(x))$ between the local functions 
induces the equality $\phi(F_W(x)) = F'_{\varphi(W)}(\phi(x))$  between the 
global functions for any update $W$.

To prove the reverse implication we need to show that every bijection $\phi$ 
can be expressed locally: 
Suppose $\N$ and $\N'$ are isomorphic and let $\varphi$ and $\phi$ match 
Definition~\ref{def_bisim}.
For all $i\in\{1,\ldots,n\}$, let $\phi_i:\B^n \to B$ be defined by $\phi_i(x) 
= \phi(x)_{\varphi(i)}$. We want to prove that $\phi_i$ does not depend on any 
other variable than $x_i$ (hence it can be rewritten as a Boolean function from 
$\B \to \B$).

Let $j\in\{1,\ldots,n\}$ and let $x$ be any configuration, then by definition,
$$\phi(F_i(x)) = \left\{\begin{array}{l} \phi(x) \text{ if $j$ is stable in 
$x$} \\ \phi(\neg{x}^j) \text{ if $j$ is unstable in $x$}  \end{array}\right.$$
and
$$F'_{\varphi(j)}(\phi(x)) = \left\{\begin{array}{l} \phi(x) \text{ if 
$\varphi(j) $ is stable in $\phi(x)$} \\ \neg{\phi(x)}^{\varphi(j)} \text{ if
$\varphi(j) $ is unstable in $\phi(x)$}  \end{array}\right.$$
The function $\phi$ is a bijection so $\phi(x) \neq \phi(\neg{x}^j)$. In the 
same time, $\phi(F_i(x)) = F'_{\varphi(j)}(\phi(x))$ and so $\phi(\neg{x}^j) 
\neq \phi(x)$ implies that $\phi(\neg{x}^j) = \neg{\phi(x)}^{\varphi(j)} (\neq 
\phi(x))$. 

So if $j\neq i$ then for all $x$, $\phi_i(\neg{x}^j) = 
\phi(\neg{x}^j)_{\varphi(i)} = (\neg{\phi(x)}^{\varphi(j)})_{\varphi(i)} 
= \phi(x)_{\varphi(j)} = \phi_i(x)$ so $\phi_i$ does not depend on $x_j$ for 
all $j\in\{1,\ldots,n\}-\{i\}$, so $\phi_i$ only depends $x_i$.

Finally, $\phi_i$ is bijective since $\phi$ is a bijection. This concludes the 
proof. 
\end{proof}

We will make great use of Theorem~\ref{thm_bisim_can} and 
Lemma~\ref{lem_bisim_cond} 
in Section~\ref{section_ex}, when we will give new isomorphism results specific 
to $\xor$-BANs.


\section{General results on $\xor$-BANs} \label{section_res}

This section presents the main theorem of this paper: a connexity result that 
characterises the shape of the ATG of any strongly connected $\xor$-BAN with an 
induced BADC of size greater than 3. 

\begin{theorem} \label{thm_reachability}
	In a strongly connected $\xor$-BAN with an induced BADC of size greater than 
3, any 
	configuration that is not unreachable can be reached from any configuration 
which 
	is not stable in a quadratic number of asynchronous updates.
\end{theorem}

This theorem tells us that the ATG of any strongly connected $\xor$-BAN which is 
not 
a cycle or a clique is characterised by (see Figure~\ref{xban_ATG}):
\begin{itemize}
 \item its fixed point(s) $S$ (if any).
 \item its unreachable configuration(s) $U$ (if any).
 \item a unique strongly connected component (SCC) of reversible transient 
configurations, reachable from any configuration of 
 	$U\setminus S$ and connected to any configuration of $S\setminus U$.   
\end{itemize}

The proof of Theorem~\ref{thm_reachability} is based on several algorithms 
that describe sequences of updates tomove from a given configuration to an 
other 
in the ATG of a $\xor$-BAN. 
We start this section by presented these algorithms. 
In a second time we briefly discuss the complexity of these 
algorithms to give an upper bound on the length of the minimal sequence of 
updates between two configurations. 
The end of the section is dedicated to general remarks about the set of fixed 
points and unreachable configurations of any BANs and helps precise the results 
of Theorem~\ref{thm_reachability}.

\begin{figure}[t!]
	\centerline{
		\scalebox{1}{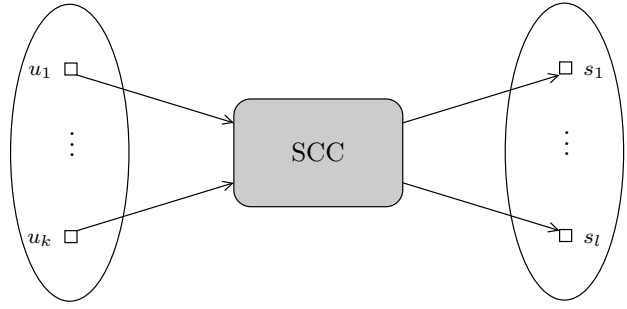}
	}
	\caption{General ATG shape of strongly connected $\xor$-BANs with an 
		induced BADC of size greater than 3.}
	\label{xban_ATG}
\end{figure}


\subsection{Proof of Theorem~\ref{thm_reachability}} \label{section_res_pf}

Let $\N$ be a strongly connected $\xor$-BAN with an induced BADC $B$ of size 
greater than 3, let $x$ be its initial (unstable) configuration, and let $x'$ 
be the configuration to reach. The idea behind the proof of Theorem 
\ref{thm_reachability} is to take advantage of the high expressiveness of 
$\xor$-BADCs and to use $B$ as a ``state generator'' that sends information 
across the network in order to set up the state of every automaton of $\N$ to 
its value in $x'$.
More precisely, if $B$ is in an unstable configuration then we will show that, 
given an automaton $i$ and a Boolean value $b$, $\N$ can always move to a 
configuration where $i$ is in state $b$ and $B$ is unstable. 

A good intuition for this is to see that, in a positive $\xor$-BADC, 
if the central automaton receives a Boolean value $1$ from one of its 
influencers then it can switch state as many times as desired by sending its 
own 
state along the opposite cycle.
To make this explicit, suppose that $x_{pred_1(o)} = 1$ then updating 
the automata along $\C_2$ will lead to a configuration where $x_{pred_2(o)} = 
x_o$ and so $f_o(x) = x_{pred_1(o)} \xor x_{pred_2(o)} = 1 \xor x_o = 
\neg{x_o}$.
Hence, in a positive network, it is possible to set any automaton $i$ to some 
state $b$, by setting $o$ to $b$ and then propagating $b$ along a path from 
$o$ to $i$. Moreover, one can ensure that this will be possible again, if 
in the end at least one of the two predecessors of $o$ is in state $1$.

To formalise this reasoning the proof of Theorem \ref{thm_reachability} 
is based on the following two lemmas.
 
 \begin{lem}\label{lem_algo_badc}
	In a $\xor$-BADC, every configuration which is not unreachable can be 
reached from 
	any other (unstable) configuration in $O(n^2)$ updates (and the bound is 
tight).
\end{lem}

\begin{proof}
First let us recall that all $\xor$-BADCs of same size $(n_1,n_2)$ are 
equivalent with respect to behavioural isomorphism. 
This means in particular that their ATGs are isomorphic and so proving that 
Lemma~\ref{lem_algo_badc} holds for positive $\xor$-BADCs is sufficent to prove 
Lemma \ref{lem_algo_badc} completely.
Hence in the following we will assume that $B$ is \emph{positive}. 
One will notice however that the proof below is easy to adjust to any 
$\xor$-BADC.

We prove Lemma~\ref{lem_algo_badc} by presenting an algorithm that explains how 
to go from one (unstable) configuration to an other (reachable) one in the ATG 
of any positive $\xor$-BADC that has at least one cycle of size greater than 
$3$. 
The algorithm can be tuned to deal with BADCs where $n_1$ and $n_2$ are both 
less than or equal to $2$ but this multiplies the number of cases that 
need to be considered and masks the general dynamics. So for the special 
case of BADCs of size $(n_1,n_2)= (1,2)$ (or vice-versa) and $(n_1,n_2)= (2,2)$ 
we prefer to prove Lemma~\ref{lem_algo_badc} by looking directly at the form of 
their ATG. These ATGs are drawn in Figure~\ref{schema_badc12} and they both 
satisfy Lemma~\ref{lem_algo_badc} as desired.
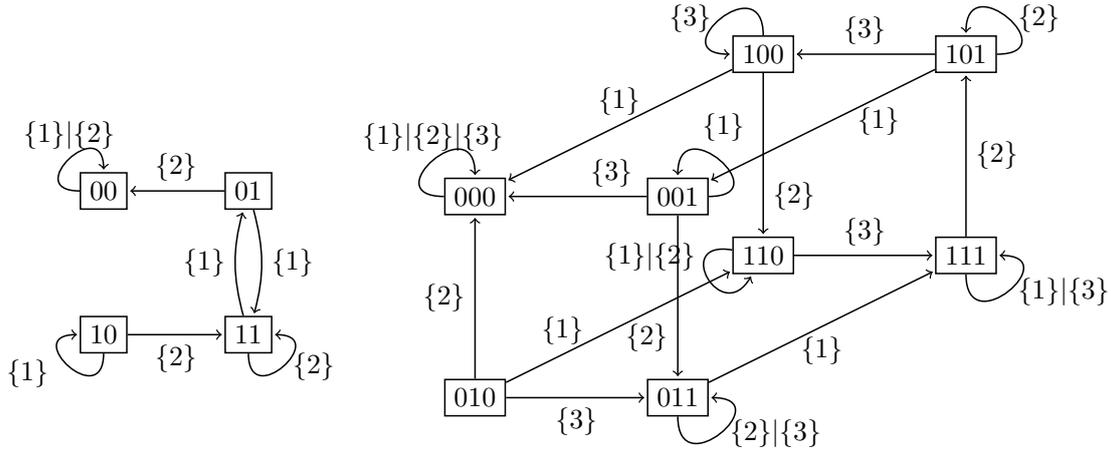
\begin{figure} 
\center
\resizebox{\textwidth}{!}{
 \begin{tikzpicture}[->,shorten >=1pt,auto,node distance=2cm, semithick]
  \tikzstyle{state}=[draw=black,text=black] 
  
  \node[state] (00) {$00$};
  \node[state] (01) [right of=00] {$01$};
  \node[state] (10) [below of= 00] {$10$};
  \node[state] (11) [right of=10] {$11$};
  \path (00) edge [out=180,in=90,looseness=5] node[above] {$\{1\}|\{2\}$} (00)
        (01) edge[bend left = 15]  node[right] {$\{1\} $} (11)
             edge  node[above] {$\{2\}$} (00)
        (10) edge node[below] {$\{2\}$} (11)
             edge [out=-90,in=-180,looseness=5] node[pos= 0.7] {$\{1\}$} (10)
        (11) edge [bend left = 15] node[left] {$\{1\}$} (01)
	     edge [out=-90,in=0,looseness=5] node[right] {$\{2\}$} (11);
	    
  \node (phantom) [below of= 10]{};
\end{tikzpicture}
 \begin{tikzpicture}[->,shorten >=1pt,auto,node distance=2.8cm, semithick]
  \tikzstyle{state}=[draw=black,text=black] 

  \node[state] (000) {$000$};
  \node[state] (001) [right of=000] {$001$};
  \node[state] (010) [below of= 000] {0$10$};
  \node[state] (011) [right of=010] {$011$};
  
  \node[state] (100) [above right of= 01]{$100$};
  \node[state] (101) [right of=100] {$101$};
  \node[state] (110)
                    [below of= 100] {$110$};
  \node[state] (111) [right of=110] {$111$};
  
  \path (000) edge [out=180,in=90,looseness=5] node[above] 
{$\{1\}|\{2\}|\{3\}$} 
(000)
        (001) edge node[left, near end] {$\{2\} $} (011)
             edge  node[above, near start] {$\{3\}$} (000)
             edge [out=0,in=90,looseness=5, near end] node[above right] 
{$\{1\}$} (001)
        (010) edge node[below] {$\{3\}$} (011)
             edge  node[left] {$\{2\}$} (000)
             edge  node[near start, above] {$\{1\}$} (110)
        (011) edge node[below] {$\{1\}$} (111)
	     edge [out=-90,in=0,looseness=5] node[right] {$\{2\}|\{3\}$} (011)
	(100) edge node[above] {$\{1\}$} (000)
	      edge node[right, near end] {$\{2\}$} (110)
	      edge[out=90,in=180,looseness=5] node[left] {$\{3\}$} (100)
	(101) edge node[below, near start] {$\{1\}$} (001)
	      edge node[above] {$\{3\}$} (100)
	      edge[out=0,in=90,looseness=5] node[right] {$\{2\}$} (101)
	(110) edge[out=170,in=240,looseness=5] node[near start, left] 
{$\{1\}|\{2\}$} (110)
	      edge node[above] {$\{3\}$} (111)
	(111) edge[out=-90,in=0,looseness=5] node[right] {$\{1\}|\{3\}$} (111)
	      edge node[right] {$\{2\}$} (101);
\end{tikzpicture}
}
\caption{The ATGs of the positive BADCs of size $(1,2)$ (left) and $(2,2)$ 
(right).} \label{schema_badc12}
\end{figure}

We now assume that $n_1 \geq 3$. The algorithm works in two steps  
(summarised in Figure~\ref{schema_algo_badc}):

\begin{enumerate}
 \item From any unstable configuration (\ie with at least one automata in state 
$1$ in the case of positive $\xor$-BADC) one can reach the highly 
expressive alternating configuration $x$ where $x_o = f_o(x)$ and $x_i = 
\neg{f_i(x)}$ for all $i \neq o$ (\ie $x_o = x_{n_1} \xor x_{n_2}$ and $x_j^k = 
\neg{x^k_{j-1}}$ for all $i^k_j \neq o$). 
This is possible for example using the following steps:
\begin{itemize}
 \item In a linear number of updates, set $x^1_{n_1}$ to $1$ and $x^2_{n_2}$ to 
$0$: 
Let $i^k_j$ be the automaton in state $1$ that is the closest to $i^1_{n_1}$ 
and update every automata on the directed path from 
$i^k_j$ to $i^1_{n_1}$.  
If $k=1$ then this simply propagates the state $1$ on every automaton from $j$ 
to $n_1$ in $\C_1$. 
If $k=2$ then this propagates the state $1$ on every automaton from $j$ to 
$n_2$ in $\C_2$ then from $1$ to $n_1$ in $\C_1$. In this second case we need 
to ensure that $x_{i_1^1} (= x_o)$ is really set to $1$ after its update, but 
this is the case since $x^1_{n_1} = 0$ and $x^2_{n_2} = 1$, hence $f_o(x) = 1$, 
by the time $o$ is updated.
Hence these first updates set $i^1_{n_1}$ to $1$.

To finish, if $x^2_{n_2}\neq 0$ (hence $x_{n_2}= 1$) then update all the 
automata of $\C_2$ from $i^1_1$ (= $o$) to $i^2_{n_2}$. When $o$ is updated 
$f_o(x) = 1\xor 1 = 0$ and so the value $0$ propagates as desired. 

 \item In a quadratic number of updates, set $\C_1$ to the alternating 
configuration where $x^1_{n_1} = 1$, \ie set $\C_1$ to $11(01)^{n_1/2-1}$ if 
$n_1$ is even and to $0(01)^{(n_1-1)/2}$ if $n_1$ is odd. This can be done as 
follows:
\begin{itemize}
 \item[]
 $\cdot$ for $j = n_1$ to $2$ do:
 \begin{itemize}
  \item \quad $\cdot$ update the automata of $\C_1$ from $1$ to $j$ 
  \item \quad $\cdot$ update the automata of $\C_2$ from $2$ to $n_2$
 \end{itemize}  
\end{itemize}

In the above algorithm, the following invariant holds: after each iteration, 
$x^1[n_1,j] 
= (10)^{(n_1-j) / 2}$ and  $x^2_{n_2} = x^1_j = x_o$, hence $f_o(x) =  
x^1_{n_1} \xor x^2_{n_2} = 1 \xor x^1_j = \neg{x^1_j}$. 
Indeed we start with $x^1_{n_1} = 1$ and $x^2_{n_2} = 0$ so by the end of the 
first iteration $x^1_{n_1} = x^2_{n_2} = x_o = 1\xor 0 = 1$.
Then, for the $j^{th}$ iteration, we start with $x^1[n_1,j+1] = 
(10)^{\frac{n_1-j+1}{2}}$ and with $f_o(x) = \neg{x^1_{j+1}}$ so we end up 
with $x^1_j = x_{n_2} = x_o = \neg{x^1_{j+1}}$ and so $x^1[n_1,j] = 
(10)^{(n_1-j) / 2}$.
 
 \item Similarly, force $\C_2$ to alternate in a quadratic number of updates 
(while preserving the alternating configuration in $\C_1$):
 \begin{itemize}
 \item[]
 $\cdot$ for $j = n_2-1$ to $2$ do:
 \begin{itemize}
  \item \quad $\cdot$ update the automata of $\C_2$ from $1$ to $j$ 
  \item \quad $\cdot$ update the automata of $\C_1$ from $n_1$ to $2$
 \end{itemize}  
\end{itemize}
 
 After each iteration, the following invariants hold: $x^2_{n_2}$ is 
unchanged, $x^1_2 = x^2_j = x_o$, $f_o(x) \neq x_o$, and $x^1[2,n_2]$ and 
$x^2[n_2,j]$ are both alternating. The first two statements are direct 
translation of the instructions. The last two require the invariant hypotheses.

By the previous point all the invariants are satisfied before entering the 
loop. Hence, right after its update $x_o \neq x^1_2$ and $x_o \neq x^2_{j+1}$. 
So after its update $x^2_j = x_o \neq x^2_{j+1}$ (hence $x^2[n_2,j]$ is 
alternating), and updating $\C_1$ in reverse order leaves it alternating. 
This also restores the fact that $x_o \neq f_o(x)$ since the state of 
$i^1_{n_1}$ has been switched with the update of $\C_1$ while the state of 
$i^2_{n_2}$ has been left unchanged.

 \item By the end of the two previous steps the system is in a configuration 
such that $f^k_j(x) = \neg{x^k_j}$ for all Automata $i_j^k$ except Automata 
$i^1_2$ and $i^2_2$.  
The last thing to do to reach a fully alternating configuration - 
where  $f_i(x) = \neg{x_i}$ for every automaton $i$ but $o$ - is thus to update 
$\C_1$ and $\C_2$ in reverse order (from $n_1$, respectively $n_2$, to $2$) and 
then update the central automaton $o$.\\
This takes a linear number of updates.
\end{itemize}
Hence, the whole sequence takes a quadratic number of updates and it results in 
one of the following alternating configurations:
 \begin{itemize}
  \item[$\cdot$] $(0(10)^{\frac{n_1-1}{2}},0(10)^{\frac{n_2-1}{2}})$ if $n_1$ 
and $n_2$ are odd, 
  \item[$\cdot$] $((10)^{\frac{n_1}{2}},1(01)^{\frac{n_2-1}{2}})$ if $n_1$ is 
even and $n_2$ is odd,
  \item[$\cdot$] $((01)^{\frac{n_1}{2}},(01)^{\frac{n_2}{2}})$ if $n_1$ and 
$n_2$ are even,
  \item[$\cdot$] $(1(01)^{\frac{n_1-1}{2}},(10)^{\frac{n_2}{2}})$ if $n_1$ is 
odd and $n_2$ is even.
 \end{itemize}

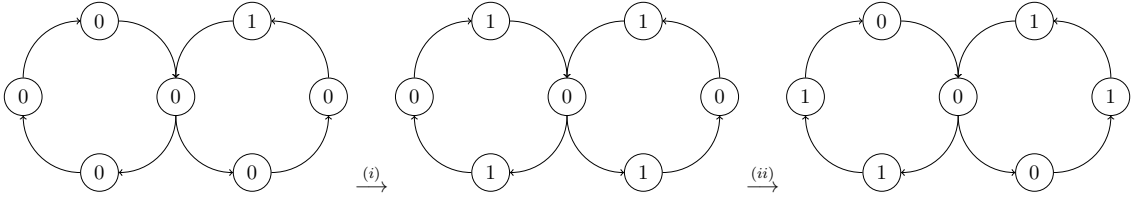
\begin{figure}
 \center
 \resizebox{\textwidth}{!}{
\begin{tikzpicture}[->,node distance = 2cm]\tikzstyle{automata}=[circle,draw, 
minimum width = .7cm];
\node[automata] (0) {};

\node[automata, below left of= 0] (1) {};
\node[automata, above left of= 1] (3) {};
\node[automata, above right of= 3] (4) {};

\node[automata, below right of= 0] (6) {};
\node[automata, above right of= 6] (7) {};
\node[automata, above left of= 7] (8) {};

\path (0) edge[bend left=45] (1)
	  edge[bend right = 45] (6)
      (1) edge[bend left=45] (3)
      (3) edge[bend left=45] (4)
      (4) edge[bend left=45] (0)
      (6) edge[bend right=45] (7)
      (7) edge[bend right=45] (8)
      (8) edge[bend right=45] (0) ;

\draw (0) node{$0$};
\draw (1) node{$0$};
\draw (3) node{$0$};
\draw (4) node{$0$};
\draw (6) node{$0$};
\draw (7) node{$0$};
\draw (8) node{$1$};
\end{tikzpicture}
$\overset{(i)}{\longrightarrow}$
\begin{tikzpicture}[->,node distance = 2cm]
\tikzstyle{automata}=[circle,draw, minimum width = .7cm];
\node[automata] (0) {};

\node[automata, below left of= 0] (1) {};
\node[automata, above left of= 1] (3) {};
\node[automata, above right of= 3] (4) {};

\node[automata, below right of= 0] (6) {};
\node[automata, above right of= 6] (7) {};
\node[automata, above left of= 7] (8) {};

\path (0) edge[bend left=45] (1)
	  edge[bend right = 45] (6)
      (1) edge[bend left=45] (3)
      (3) edge[bend left=45] (4)
      (4) edge[bend left=45] (0)
      (6) edge[bend right=45] (7)
      (7) edge[bend right=45] (8)
      (8) edge[bend right=45] (0) ;

\draw (0) node{$0$};
\draw (1) node{$1$};
\draw (3) node{$0$};
\draw (4) node{$1$};
\draw (6) node{$1$};
\draw (7) node{$0$};
\draw (8) node{$1$};
\end{tikzpicture}
$\overset{(ii)}{\longrightarrow}$
\begin{tikzpicture}[->,node distance = 2cm]
\tikzstyle{automata}=[circle,draw, minimum width = .7cm];
\node[automata] (0) {};

\node[automata, below left of= 0] (1) {};
\node[automata, above left of= 1] (3) {};
\node[automata, above right of= 3] (4) {};

\node[automata, below right of= 0] (6) {};
\node[automata, above right of= 6] (7) {};
\node[automata, above left of= 7] (8) {};

\path (0) edge[bend left=45] (1)
	  edge[bend right = 45] (6)
      (1) edge[bend left=45] (3)
      (3) edge[bend left=45] (4)
      (4) edge[bend left=45] (0)
      (6) edge[bend right=45] (7)
      (7) edge[bend right=45] (8)
      (8) edge[bend right=45] (0) ;

\draw (0) node{$0$};
\draw (1) node{$1$};
\draw (3) node{$1$};
\draw (4) node{$0$};
\draw (6) node{$0$};
\draw (7) node{$1$};
\draw (8) node{$1$};
\end{tikzpicture}
}
 \caption{An example for Lemma~\ref{lem_algo_badc} with a $(4,4)$ positive 
$\xor$-BADC, from configuration $(0000,0001)$ to configuration 
$(0110,0011)$. The algorithm works in two steps: first setting $B$ in 
a fully alternating configuration, then updating each automaton according to 
its targeted state.}
\label{schema_algo_badc}
\end{figure}

 \item Let $x$ denote the resulting alternating configuration, then any 
configuration $x'$ with at least one automaton $i^k_j$ in 
stable state (\ie such that $x'^k_j = f^k_j(x')$) is reachable from 
$x$.

 Indeed,  $x_o = f_o(x)$ and for all $i \neq o$, $x_i = \neg{f_i(x)}$ so in a 
linear number of updates we can move from the configuration $x$ to the 
configuration $\hat{x}$ where $\hat{x}_o =  x'_o$ and  $\hat{x}_i = 
\neg{f_i(\hat{x})}$ for all $i \notin \{i_k^j, o\}$. This is achieved by 
following instructions: 
\begin{itemize}
 \item[] $\cdot$ if $i_k^j \neq o$ and $x'_o \neq x_o$
 \begin{itemize}
   \item[] \quad $\cdot$ update $o$ and the automata from $n_k$ to $j$ in 
$\C_k$.
 \end{itemize}
\end{itemize}

Then, reaching $x'$ from $\hat{x}$ is straightforward: one simply needs to 
switch the state of the automata when necessary:
\begin{itemize}
 \item[] $\cdot$ for $j = n_1$ to $2$ (in $\C_1$): update the automaton $i^1_j$ 
if $\hat{x}^1_j \neq x'^1_j$;
 \item[] $\cdot$ for $j = n_2$ to $2$ (in $\C_2$):  update the automaton 
$i^2_j$ 
if $\hat{x}^2_j \neq x'^2_j$;
 \item[] $\cdot$ update the automaton $i_k^j$.
\end{itemize}
These updates are efficient since for all $i\notin \{i_k^j, o\}$, if $\hat{x}_i 
\neq x'_i$ then $x'_i = \neg{x_i} = 
f_i(\hat{x})$, 
which is the value returned by the update of $i$. Then, by definition of 
$\hat{x}$, automaton $o$ already has the right 
state. And, finally, by definition of $i_k^j$,  $x'^k_j = f^k_j(x')$, which is 
the value returned by $f_i$ after all the other automata have been updated.

The second sequence takes a linear number of steps, so the whole sequence 
remains quadratic. 
This bound is tight since going from the configuration
$x = (10^{n_1-1},10^{n_2-1})$ to a configuration $x'$ where 
$x'_i = \neg{f_i(x')}$ for all Automata $i\neq o$ (as for example the 
configuration $x' = (0(10)^{\frac{n_1-1}{2}},0(01)^{\frac{n_2-1}{2}}$ if $m$ 
and 
$n$ are odd) requires at least $\sum_{j=1}^{n_1} j + \sum_{j=1}^{n_2} j = 
\frac{n_1(n_1-1)}{2} + \frac{n_2(n_2-1)}{2}$ updates, which is in 
$\theta((n_1+n_2)^2)$.
\end{enumerate}
\end{proof}

\begin{remark}
Note that if synchronous transitions are allowed, then every configuration is 
reachable from any unstable configuration. Indeed, the above algorithm says 
that it is immediate if the target configuration is not unreachable, but it 
also tells us that if $x$ is unreachable, one can still reach the configuration 
$\hat{x} = \neg{x}^{\C_1-\{o\}}$, since in that case 
$f_o(\neg{\hat{x}}^{o}) = (\neg{\hat{x}}^{o})_{n_1}^1 \xor 
(\neg{\hat{x}}^{o})_{n_2}^2
= \hat{x}_{n_1}^1 \xor \hat{x}_{n_2}^2 
= \neg{x_{n_1}^1} \xor x_{n_2}^2 
= \neg{x_{n_1}^1 \xor x_{n_2}^2 }
= \neg{f_o(\neg{x}^o)}
= x_o = \hat{x}_o$ (we assume $B$ positive) and so $\hat{x}$ is not 
unreachable. 

Then for every automaton $i_j^1$ of $\C_1-\{o\}$, 
$f_j^1(\hat{x}) = f_j^1(\neg{x}^{\C_1-\{o\}})
= (\neg{x}^{\C_1-\{o\}})_{j-1}^1
=  \neg{x_{j-1}^1} 
= \neg{f_j^1(\neg{x}^{i_1^j})} 
= x_{j-1}^1$, so the synchronous update of $\C_1-\{o\}$ 
changes the configuration of the system from $\hat{x}$ to $x$.
\end{remark} 

\begin{lem}\label{lem_influence}
	In a $\xor$-BAN $\N$, if $i$ and $j$ are two automata such that there is a 
path 
	from $i$ to $j$, then for any configuration $x$ such that $i$ is unstable 
	in $x$ there exists a configuration $x'$ reachable from $x$ such that $j$ is 
unstable in $x'$. 
\end{lem}
\begin{proof}
The proof is based on the fact that, in a $\xor$-BAN, making a stable 
automaton become unstable can simply be achieved by switching the state of one 
of its incoming neighbours (because the state of an automaton depends on the 
parity of the number of its incoming neighbours in state $1$).

So let $i$ and $j$ be two automata as described in Lemma~\ref{lem_influence}, 
let $p = i_0,i_1,\ldots, i_k$ be a shortest path (in the interaction graph of 
$\N$) from $i = i_0$ to $j = i_k$ and let $i_\ell$ denotes the last 
automaton in $p$ that is unstable. Then updating along $p$ from $i_\ell$ to 
$i_{k-1}$ (so that nothing happens if $\ell = k$, \ie if $j$ is unstable) will 
lead to a configuration where $j$ is unstable. 
This is straightforward from the remark above. The only subtlety is the choice 
of the path which must ensure that the update of one automaton only affects the 
next automaton on the path but not the automata after it, and this is true if 
one takes a shortest path. 
\end{proof}

Putting things together we can now describe the algorithm underlying the proof 
of Theorem \ref{thm_reachability}:
\begin{proof}
Let $B$ be an induced BADC of size greater than $3$ in the BAN $\N$ and let $x$ 
and $x'$ respectively be the initial configuration and the target configuration 
described in Theorem~\ref{thm_reachability}. The configuration $x$ is not 
stable so, by Lemma~\ref{lem_influence}, it is possible to go from $x$ to a 
configuration $y$ where one automaton of $B$, hence $B$, is not stable. Then, 
using Lemmas~\ref{lem_algo_badc} and~\ref{lem_influence}, we claim that it is 
possible to set the state of every automata $i$ outside of $B$ to its value in 
$x'$ while keeping $B$ in an unstable configuration.

The idea is as follows: let $i$ be an automaton that is not in $B$ and let $p = 
i_0i_1\ldots i_k$ be a shortest path (in the interaction graph of $\N$) from 
$B$ 
to $i_k = i$. Then, applying the algorithm from Lemma~\ref{lem_algo_badc}, 
we know how to reach a configuration where $i_0$ is unstable and so, using the 
algorithm from Lemma~\ref{lem_influence}, we know how to reach a configuration 
where $i$ is unstable. From this configuration we can set the state of $i$ to 
$x'_i$ by updating $i$ if necessary. 

So, if we can guarantee that this process preserves the instability in $B$, 
then 
we can use it repetitively on every automaton outside of $B$ to reach a 
configuration where $B$ is unstable and where all automata outside of $B$ are 
in 
the state specified by $x'$. Once this is done we only need to set $B$ to its 
right value to reach $x'$ and, since $B$ is unstable, this can be done by using 
the algorithm from Lemma~\ref{lem_algo_badc}. 

In fact, setting the automata outside of $B$ to their state in $x'$ cannot be 
done in any order. Indeed, the algorithm from Lemma~\ref{lem_influence} 
requires to switch the state of some automata outside of $B$ (namely the one 
along the path from $B$ to the automaton to be set up). Hence we need to 
guarantee that the automata that have already been treated are not switched 
again while processing the other automata. A way to ensure that is to compute a 
breadth first search tree of root $B$ and to treat the automata in the order 
given by the tree from the leaves to the root, using the branches of 
the tree as the paths from $B$ to the automata to be treated.
An example of such ordering is given in Figure~\ref{fig_bft}. 

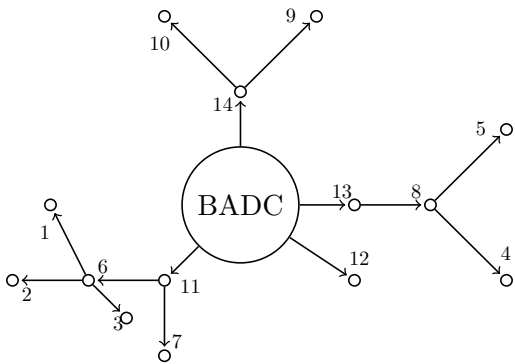
\begin{figure*}[h]
  \center
\begin{tikzpicture}[->,shorten >=1pt,auto,node distance=1.5cm, semithick]
\tikzstyle{vertex}=[draw = black,circle,text=black,inner sep=1.5pt]
\tikzstyle{label}= [very near end,scale=0.7]

  \node[vertex] (badc) at (3,0) {\ BADC \ };
  \node[vertex] (1) at (2,-1) {};
  \node[vertex] (2) at (3,1.5) {};
  \node[vertex] (3) at (2,2.5) {};
  \node[vertex] (4) at (4,2.5) {}; 
  \node[vertex] (5) at (2,-2) {};
  \node[vertex] (6) at (1,-1) {};
  \node[vertex] (7) at (1.5,-1.5) {};
  \node[vertex] (8) at (4.5,0) {};
  \node[vertex] (9) at (5.5,0) {};
  \node[vertex] (10) at (6.5,1) {};
  \node[vertex] (11) at (6.5,-1) {};
  \node[vertex] (12) at (4.5,-1) {};
  \node[vertex] (13) at (0,-1) {};
  \node[vertex] (14) at (0.5,0) {};
  
  \path (badc) edge node[label]{11} (1) edge node[label]{12} (12) edge 
node[label]{14} (2) edge node[label]{13} (8);
  \path (1) edge node[label]{7} (5) edge node[label,above]{6} (6) (6) edge 
node[label,below]{3} (7) edge node[label]{2} 
(13) edge node[label]{1}  (14);
  \path (2) edge node[label]{10} (3) edge node[label]{9} (4);
  \path (8) edge node[label]{8} (9) (9) edge node[label]{5} (10) edge 
node[label]{4} (11);
\end{tikzpicture} 
\caption{Example of update order using a breadth first tree.} \label{fig_bft}
\end{figure*}

Moreover the use of Lemma~\ref{lem_algo_badc} at the end of the update
sequence requires that the restriction of $x'$ to $B$ is not unreachable for 
$B$ (\ie for $B$ viewed as a $\xor$-BADC whose local transition functions are 
fixed by its surrounding environment in $x'$).
If this is not the case, we have to get around the problem by using the same 
kind of trick that the one used in the second step of the proof of 
Lemma~\ref{lem_algo_badc} \textendash when the stable state of the target 
configuration is not the central node $o$:

Let $i$ is an automaton of $\N$ such that $f_i(\neg{x'}^i) = x'_i$ ($i$ exists 
since $x'$ is reachable), and let $p= i_0\ldots i_k $ be a shortest path from 
$i = i_0$ to $B$. Then we first reach the configuration $\hat{x}$ such that 
\emph{(i)} $\hat{x}_j = x'_j$ if $j\notin p$, \emph{(ii)} $i_k (\in B)$ is such 
that the restriction of $\hat{x}$ to $B$ is reachable for $B$, and \emph{(iii)} 
the state values of the automata in $p$ are ``alternating'' in such a way that 
if we set up the state of the automata of $p$ to their value in $x'$ from $i_k$ 
to $i_1$ then every time an automaton $i_\ell$ is about to be set up, its 
predecessor in $p$ must be unstable so as to enable $\ell$ to switch state if 
necessary. 

With such conditions it is easy to go from $\hat{x}$ to $x'$: one only 
needs to set up $p$ back up. As described in condition \emph{(iii)}, every 
automaton in $p-\{i_0\}$ will be able to switch state in turn if necessary, 
then, in the end, if $i_0$ is not already in state $x'_{i_0}$ it will still be 
able to switch to the right state since $f_i(\neg{x'}^i) = x'_i$ by assumption.

The configuration $\hat{x}$ described above can be computed inductively by 
taking the $k^{th}$ iteration, $\hat{x}^{k}$, of:  
\begin{enumerate}
 \item $\hat{x}^{0} = x'$
 \item for $\ell>0$, $\hat{x}^{\ell}$ is inductively defined by:  
$\hat{x}^{\ell}_j = \hat{x}^{\ell-1}_j$ for all $j \notin \{i_{\ell-1},
i_{\ell} \}$,
$\hat{x}^{\ell}_{i_\ell} = \neg{x'_{i_\ell}}$, \quad and
$\hat{x}^{\ell}_{i_{\ell-1}}$  is the solution of  
$f_{i_{\ell-1}}(\hat{x}^{\ell}) = \neg{\hat{x}^{\ell}_{i_{\ell-1}}} $
\end{enumerate}

Finally, to conclude the proof above, we still need to precise the way of using 
the algorithm from Lemma~\ref{lem_influence} that ensures that the instability 
of $B$ is preserved by the updates outside of $B$.

So let $z^0$ be the current configuration and let $p=i_0,\ldots,i_k$ be the 
path from $B$ to the automaton to be set up. Moreover, let $j\neq i_0$  be an 
influencer of $i_0$ in $B$.

By assumption, $B$ is unstable in $z$, so one can use Lemma~\ref{lem_algo_badc} 
to put $\N$ in a configuration $z^0$ where $i_0$ is unstable, and such that:
$$z^0_j = \left\{\begin{array}{ll}
\neg{f_j(\neg{z^0}^{i_1})} &\text{ if $i_1$ is an influencer of } i_0 \
(i_1\in I(i_0)) \\
\neg{f_j(\neg{z^0}^{\{i_0,i_1\}})} &\text{ if $i_1$ is not an influencer of } 
i_0 \ (i_1\notin I(i_0))
\end{array}\right. .
$$
Actually, one can only guarantee that this is possible if $B$ has one cycle of 
size at least $3$, which enables to ask for a third automaton (different from 
$i_0$ and $j$) to be stable in $z^0$, making $z^0$ reachable for $B$.

From there one can start applying Lemma~\ref{lem_influence}: \\
Let $i_\ell$ be the last automaton in $p$ that is unstable. If $\ell\leq 
1$, then start updating $p$ from $i_\ell$ to $i_1$. This leaves $\N$ in a 
configuration $z^1$ such that $B$ is unstable. Indeed:
\begin{itemize}
 \item either nothing happened ($\ell>1$) and so $B$ is still unstable (because 
$i_0$ is unstable in $z^0$ for example).

 \item or Automaton $i_1$ is the only to have been updated and so:\\
 \emph{(i)} if $i_1\in I(i_0)$, then $z^1_j = z^0_j = 
\neg{f_j(\neg{z^0}^{i_1})} 
= \neg{f_j(z^1)}$ and so $j$ is unstable in $z^1$; \\
 \emph{(ii)} if $i_1\notin I(i_0)$  then the neighbourhood of $i_0$ has not 
changed so $i_0$ is still unstable in $z^1$.

\item or both Automata $i_0$ and $i_1$ have been updated and so:\\
\emph{(i)} if $i_0\notin I(i_0)$ ($i_0$ has no self loop) and if 
$i_1 \in I(i_0)$, then $i_0$ is still unstable (since it has changed 
and an odd number of its incoming neighbours have changed too);\\
\emph{(ii)} if $i_1 \notin I(i_0)$ then as previously $z^1_j = z^0_j 
= \neg{f_j(\neg{z^0}^{\{i_0,i_1\}})} =  \neg{f_j(z^1)}$ and so $j$ is unstable 
in $z^1$; \\
\emph{(iii)} if $i_0 \in I(i_0)$ then $i_0$ is not an influencer of $j$ 
(because $B$ is an induced BADC of size $3$ and $j$ has been chosen to be the 
predecessor of $i_0$ different from $i_0$) 
so $f_j(\neg{z^0}^{\{i_1\}}) =  f_j(\neg{z^0}^{\{i_0,i_1\}})$, so $z^1_j = 
\neg{f_j(\neg{z^0}^{\{i_0,i_1\}})}$ 
which means as previously that $j$ is unstable in $z^1$.
\end{itemize}
Now, let $\ell' = \max(2,\ell)$, $\ell'$ is the last automaton of $p$ 
to be unstable in $z^1$.
Then, again, $B$ is unstable in $z^1$ so we can use Lemma~\ref{lem_algo_badc} 
to 
reach a configuration $z^2$ such that 
$z^2_{i_0} = \neg{f_{i_0}(\neg{z^1}^{\{i_{\ell'},\ldots, i_{n-1}\}})}$ 
and $z^2_i = z^1_i$ for all $i \notin B$. 
Moreover, since $p$ was chosen to be a shortest path, no automata in $B$ 
influence the automata of index greater than $2$ in $p$. So the last automaton 
of $p $ that is unstable in $z^2$ is still $i_{\ell'}$. 

Hence we can finish running the algorithm of Lemma~\ref{lem_influence} (by 
updating the automata along $p$ from $i_{\ell'}$ to 
$i_{n-1}$) and be sure that this leads to a configuration where $i_{n-1}$ is 
unstable. We also know that in this 
configuration $B$ is unstable since $i_0$ has state $ 
\neg{f_{i_0}(\neg{z^1}^{\{i_{\ell'},\ldots, i_{n-1}\}})}$. 

This last remark concludes the proof of Theorem~\ref{thm_reachability}.
\end{proof}


\subsection{Algorithmic complexity}
The algorithm described above is quadratic in the worst case. However, its  
complexity highly depends on the structure of the network and/or the final 
configuration $x'$. 
For example, if every automaton in $\N$ is at bounded distance from the 
central node of an induced BADC of size greater than $3$, then this algorithm 
becomes linear in $n$. 
Similarly, since the number of passes that are needed along a path depends on 
the number of alternating states (\ie $01$ or $10$ patterns) along this path in 
$x'$, then if this number is less than a constant in any path the algorithm 
will also run in linear time. This is especially the case when $x'$ is a 
fixed point of $\N$ and so every transient configuration can reach every stable 
state in a linear number of updates.

Finally we need to point out the fact that this algorithm does not always 
provides the most efficient sequence of updates (for example it does not take 
into account the starting configuration) hence the complexity of this algorithm 
is only an upper bound on the length of the shortest path between two 
configurations. However, let us notice that this bound can sometimes be 
reached, as when one move from configuration $10^{n-1}$ to configuration 
$(10)^{n/2}$ in a positive $\xor$-BADC of size $n$. These 
considerations on $01$ patterns echo to the notion of {\em 
expressiveness} defined for the monotonic case in \cite{mrrs15}.

\subsection{Fixed points and unreachable configurations} \label{section_res_rmk}

According to the definition, a configuration $x$ is a fixed point for a mode if 
it 
has no outgoing arcs but self-loops in the transition graph associated to 
this mode. In the asynchronous update mode this means that for all $i$ in $V$, 
$f_i(x) = x_i$. Hence, in a fixed point, the state of the automata along a nude 
path is completely determined by the head of this nude path. This leads to the 
following bound on the number of possible fixed points, that is related 
to the set of works~\cite{richard1,richard3,richard2}.

\begin{prop}\label{lem_n_fixpoint}
	In any BAN $\N$, the maximum number of fixed points in the asynchronous mode 
$A$ is 
	$2^k$, where $k$ is the number of automata $i$ such that $\pi_i$ is of 
length $0$
	(\ie $i$ is an ``intersection node'' in some interaction graph of $\N$).
\end{prop}

\begin{proof}
 It is enough to note that a configuration $x$ is stable in $A$ only if every 
automata along a  
nude path share the same state value in $x$. In other words, $x$ is completely 
determined by the states of 
the intersection nodes of $\N$.
\end{proof}

This bound is rough and we believe that it is possible to lower it for 
subclasses of networks. However, if we define the \emph{contraction}  of a 
network to be the network obtained by removing any automaton $i$ whose incoming 
maximal nude path $\pi_i$ has length greater than 1 and replacing the variable 
$x_i$ by the variable associated to the head of $\pi_i$ in the remaining local 
functions, then any BAN whose contraction results in the trivial network 
$\{f_i(x) = x_i\}_{i\in V}$ reaches the bound of $2^k$ fixed points.

Also, notice that in the asynchronous mode, the unreachable configurations of a 
network 
$\N=\{f_i\}_{i=1}^n$ are exactly the fixed points of the \emph{reverse} network 
$\N^R=\{f^R_i\}_{i=1}^n$ defined by $f^R_i(x) = \neg{f_i(\neg{x}^i)}$. 
$\N$ and $\N^R$ are of the same type hence the maximum number of fixed points 
for the type of $\N$ will also be its maximum number of unreachable 
configurations. Moreover, this implies that if all the networks of a given type 
are behaviourally isomorphic then the number of unreachable configurations and 
the number of fixed points will be equal. 
These remarks will be illustrated by the description of the ATGs of 
$\xor$-Flowers and $\xor$-Cycle Chains presented in the next section.


\section{Study of some specific $\xor$-BANs} \label{section_ex}

We now give a complete characterisation of two specific types of $\xor$-BAN: 
the $\xor$-BA Flowers and the $\xor$-BA Chains.
For each of these two types of BANs, we describe their behavioural isomorphism 
classes and give their number of fixed points and unstable configurations.
This illustrates the results of Section \ref{section_res}, and introduces 
a new method for computing isomorphism classes through the use of rewriting 
on the interaction graph of the BANs: two BANs will be equivalent if one 
can be rewritten into the other.


\subsection{$\xor$-BA Flowers} \label{section_ex_BAF}

A $\xor$-BA Flower ($\xor$-BAF) with $m$ petals is defined as a set of $m$ 
cycles that intersect at a unique automaton $o = i^1_1 = \ldots = i^k_1$ 
($\xor$-BADCS correspond to the case $m=2$).
The following states that there are at most two isomorphism 
classes for a given type of flower, \ie for a given number 
of petals $m$ and size $(n_1,\dots,n_m)$.
\begin{prop}
The set of $\xor$-BAF with $m$ petals of size $(n_1,\ldots,n_m)$ admits one 
isomorphism class if $m$ is even and two if $m$ is odd.
\end{prop}
\begin{proof}
 Similarly to what is done in Section \ref{section_def} for $\xor$-BADCs, we 
restrict our study to canonical $\xor$-BAFs, that are $\xor$-BAFs such 
that the only negative literals are in the local function of $o$ 
(Theorem~\ref{thm_bisim_can}). 
According to the identity $b_1 \xor b_2 = \neg{b_1} \xor \neg{b_2}$ that holds 
for every Boolean values $b_1$, $b_2$, the sign of any pair of negative 
literals cancel in $f_o$. So there are at most two isomorphism classes for a 
given type of flower: 
the positive one, where $f_o$ has only positive literals (which thus 
corresponds to BAFs with an even number of negative cycles), 
and the negative one where $f_o$ has exactly one negative literal (which thus 
corresponds to BAFs with an even number of negative cycles).

Moreover, when $m$ is even, the bijection $\phi(x) = \neg{x}^V$ over the set of 
configurations actually defines an isomorphism between the ATGs of the negative 
and the positive $\xor$-BAF of same type. Therefore, for $m$ even, the negative 
and positive classes coincide.
On the contrary, when $m$ is odd, the two classes remain distinct since, in 
particular, they do not have the same number of fixed points, as this is shown 
in the next proposition (\ref{lem_ex_BAF_stable}).
\end{proof}

\begin{prop}\label{lem_ex_BAF_stable}
  A positive $\xor$-BAF with $m$ petals has a unique stable configuration, 
$0^{n}$, if $m$ is even and two stable configurations, $0^{n}$ and $1^{n}$, if 
$m$ is odd. A negative $\xor$-BAF (with an odd number of petals) does not have 
any fixed point.
\end{prop}
\begin{proof}
There are several ways to compute the fixed points of a $\xor$-network.
network, 
One way is to fix the state of one automaton and to propagate the information 
that this choice implies on the state of the other automata in the network, 
making new choices when necessary, until having completely fixed the 
configuration or 
until reaching a contradiction.

For example, in a positive $\xor$-BAF $\mathcal{F}$ with an even number of 
petals, any configuration $x$ that contains an automaton $i$ in 
state $1$ is unstable.
Indeed suppose for the sake of contradiction that $x$ is stable, then $o$, and 
so every automata in $\mathcal{F}$, are in state $1$ (because updating from $o$ 
to $i$ implies that $x_i=x_o$), so $x = 1^n$. But $1^n$ is not stable since 
$f_o(x) = \Xor_{k=1}^m 1 = 0$. This is a contradiction.
Similarly we prove for a negative $\xor$-BAF with an odd number of petals, 
if a configuration contains an automaton in state $0$, respectively an 
automaton 
in state $1$, then it cannot be stable, and so the BAF has no fixed points.
\end{proof}

The results above allow us to fully characterise the ATG, $G^A_\mathcal{F}$, of 
any $\xor$-BAF, $\mathcal{F}$, of a given type:
\begin{itemize}
 \item if $\mathcal{F}$ has an even number of petals then 
$\mathcal{F}$ and $\mathcal{F}^R$ are in the same isomorphism class 
(the unique positive class). 
Hence, $G^A_\mathcal{F}$, has exactly one unreachable configuration, one 
fixed point, and one SCC of $2^{n}-2$ transient configurations.
 \item if $\mathcal{F}$ has an odd number of petals then  
$G^A_\mathcal{F}$ can have four different shapes depending on the size of $\F$ 
and depending on its isomorphism class.
Indeed, $\mathcal{F}$ and $\mathcal{F}^R$ are isomorphic if and only if 
$\mathcal{F}$ has an even number of petals of even sizes and a self loop, or if 
it has an odd number of petals of even sizes and no self loop. 
Hence $G^A_\mathcal{F}$ has one of the following forms:
\emph{(i)} a unique attractor of size $2^n$ if $\mathcal{F}$ 
and $\mathcal{F}^R$ are in the negative class ; 
\emph{(ii)} two unreachable configurations, two fixed points, and one SCC 
of $2^n-4$ transient configurations if $\mathcal{F}$ and $\mathcal{F}^R$ 
are in the positive class;
 \emph{(iii)} two fixed points, and one SCC of $2^n-2$ transient 
configurations if $\mathcal{F}$ is in the positive class and $\mathcal{F}^R$ 
in the negative class ;
\emph{(iv)}  two unreachable configurations and one attractor of size 
$2^n-2$ if $\mathcal{F}$ is in the negative class and $\mathcal{F}^R$ in the 
positive class.
\end{itemize}


\subsection{$\xor$-BAC Chains} \label{section_ex_BACC}
A $\xor$-BAC Chain ($\xor$-BACC) of length $m$ is described by a set of $m$ 
cycles, $\C_k$, and $m-1$ intersection automata, $o_k$, such that for all 
$1\leq k < m$, $\C_k$ intersects $\C_{k+1}$ at a unique point
$o_k= i_1^k = i_{\ell_k}^{k + 1}$. 
As previously, we characterise the isomorphism classes and the ATG of this 
type of BANs.

\begin{figure}[t!]
	\centerline{\includegraphics[scale=.25]{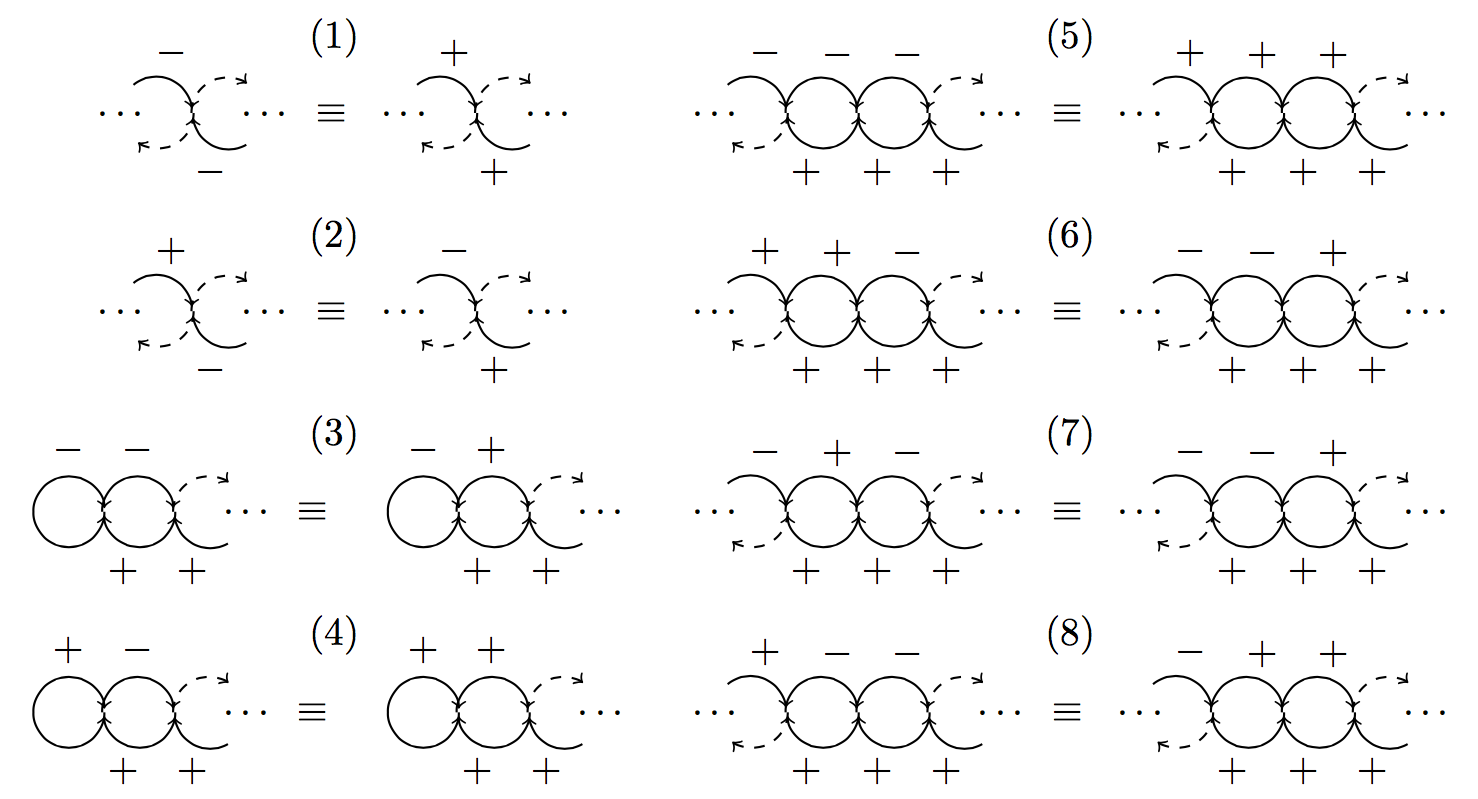}}
	\caption{Table of $\xor$-equivalences.}
	\label{eq_table}
\end{figure}


\subsection*{\textbf{Isomorphism classes}}
\begin{prop} \label{lem_bacc_classes}
The set of $\xor$-BACCs of length $m$ and size $(n_1,\ldots,n_m)$ admits one  
isomorphism class if $m-1$ is not a multiple of $3$ 
and two if $m-1$ is a multiple of $3$.
\end{prop}

As in the case of $\xor$-BAFs, the proof of Proposition~\ref{lem_bacc_classes} 
is done in two steps.
\begin{itemize}
 \item[] \textbf{Point 1.} We first show that the set of $\xor$-BACCs of a 
given type $(n_1,\ldots,n_m)$ is divided into two classes: the positive class 
and the negative class, which respectively corresponds to the 
isomorphism class of the BACC $(n_1,\ldots,n_m)$ where all path are positive, 
and the isomorphism class of the BACC $(n_1,\ldots,n_m)$  
where all paths are positive except the one from $i^1_2$ to $i^1_1$ that is 
negative.
 \item[] \textbf{Point 2.} We then prove that, in fact, when $m-1$ is not a 
multiple of $3$, 
the two classes coincide since the positive BACC and the negative BACC are 
isomorphic in this case.
\end{itemize}

The proof of these two points is based on the equivalences 
presented in Figure~\ref{eq_table}. Each pattern of these equivalences 
describes a subnetwork where every intersection automaton is a 
$\xor$-automaton and every arc represents a signed path of arbitrary 
length (hence containing possibly several automata).
These equivalences have to be understood as follows: given a BAN such that the 
left pattern of an equivalence appears in its interaction graph, then this BAN 
is behaviourally isomorphic to the BAN that has the 
same interaction graph except that the left pattern has been replaced by the 
right pattern of the equivalence, no matter what the outgoing dashed arcs are 
and no matter their number. 
In other words, Figure~\ref{eq_table} presents a set of interaction graph 
rewriting rules that produce equivalent networks according to the 
(behavioural) isomorphism relation.

The following lemma (\ref{lem_equivalence}) says in particular that it is 
enough to prove that the interaction graphs of two BANs can be rewritten one 
into an other using the equivalences from Figure~\ref{eq_table}, to prove that 
the two corresponding BANs are equivalent.

\begin{lem} \label{lem_equivalence}
 The interaction graph rewriting rules depicted in 
Figure~\ref{eq_table} preserve the behavioural isomorphism equivalence.
\end{lem}

\begin{proof}
Equivalences \emph{(1)} and \emph{(2)} only translate the well known 
identities $b_1 \xor b_2 = \neg{b_1} \xor \neg{b_2}$ and 
$\neg{b_1} \xor b_2 =  b_1 \xor \neg{b_2}$ for any Boolean values $b_1$ and 
$b_2$. 

The proofs of the other equivalences are a bit longer but do not present any 
difficulty. We now present a proof for the third equivalence, proofs for 
the other equivalences are similar:

Let $\N=\{f_i\}$ and $\N'=\{f_i'\}$ be two BANs whose interaction graphs only 
differ by the pattern shown in Equivalence~\emph{(3)}. We denote by $\C_1$, 
$\C_2$ the two cycles of the pattern. Similarly, $o_1$ and $o_2$ denote the 
intersection automata and $\C_2^u$ denotes the upper half-cycle of $\C_2$.
We are going to prove that $\N$ and $\N'$ are isomorphic by giving a bijection 
$\varphi:V\to V'$ and a set of local bijections $\{\phi_i:\B \to \B\}_{i\in V}$ 
satisfying the conditions from Lemma~\ref{lem_bisim_cond}.

Let $\varphi$ be the identity over the set of automata  and let $\phi_i = 
neg_\B$ if $i\in \C_1\cup \C^u_2\cup \{o_1\}$ and $\phi_i = id_\B$ 
otherwise. 
We need to check that $\phi_i(f_i(x)) = f'_i(\phi_i(x))$ for all automata 
$i$ in the network. 
This is immediate for all automata that do not belong to $\C_1\cup \C^u_2\cup 
\{o_1,o_2\}$ since for these automata we have used the identity everywhere. 
Now, if $i\in \C_1\cup \C^u_2$, then $\phi_i(f_i(x)) = \phi_i(pred(i)) 
= \neg{pred(i)}=\phi_{pred(i)}(pred(i)) = f'_i(\phi(x))$ and so the identity 
holds. 
Finally it remains to check that the identity holds for Automata $o_1$ and 
$o_2$. This is the case since:  
\begin{enumerate}
 \item $\phi_{o_1}(f_{o_1}(x)) = \phi_{o_1}(\neg{pred_1(o_1)}\xor pred_2(o_1))$
 $= \neg{\neg{pred_1(o_1)}\xor pred_2(o_1)} $
 
$= \neg{\phi_{pred_1(o_1)}({pred_1(o_1))} \xor \phi_{pred_2(o_1)}(pred_2(o_1))} 
= f'_{o_1}(\phi(x))$,
\item[] and
\item $\phi_{o_2}(f_{o_2}(x)) = \phi_{o_2}(\neg{pred_1(o_2)}\xor pred_2(o_2))$
$ = \neg{pred_1(o_2)}\xor pred_2(o_2)$

$=\phi_{pred_1(o_2)}({pred_1(o_2))}\xor \phi_{pred_2(o_2)}(pred_2(o_2)) $
$= f'_{o_2}(\phi(x))$.
\end{enumerate} 
\end{proof}
Using the equivalence of Lemma~\ref{lem_equivalence} we can now finish the 
proof of Proposition~\ref{lem_bacc_classes}. 
As mentioned above, we first show that the interaction graph of any $\xor$-BACC 
can be rewritten into an interaction graph with at most one negative path from 
$i^1_2$ to $o_1 (= i^1_1)$. This proves that there are at most two isomorphism 
classes for a given $\xor$-BACC type, the positive one and the 
negative one.
Then we prove that if $m-1$ is not a multiple of $3$ this negative path  
can actually be removed by an other sequence of rewrites, hence proving that 
the two classes are equal in this case.

\begin{proof}(Point 1.)
As usually we focus on canonical BANs, since this already reduces the number 
of cases to consider.
 Then using Equivalences \emph{(1)} and \emph{(2)} from Figure~\ref{eq_table} we 
rewrite the interaction graph of any of the canonical $\xor$-BACC into 
interaction graphs where the only negative paths are paths from $o_i$ to 
$o_{i+1}$ for $i\in\{1,\ldots,m-2\}$, that is, the only negative paths are ``on 
the top''. 

Then, inductively on the negative path of higher index (the negative path 
from $o_{i}$ to $o_{i+1}$ such that $i$ is maximal), we use the equivalences 
\emph{(5)}, \emph{(6)}, \emph{(7)} and \emph{(8)} from left to right to lower 
this index by at least one after every rewrite. We stop the rewriting when $i = 
0$ or when there are no negative paths left. In other words we do an inductive 
sequence of rewrites on the ``right most'' negative path so as to ``push'' 
this path to the left until reaching the end of the chain or making it 
disappear. An example of such a rewrite sequence is presented in 
Figure~\ref{fig_rewrite_ex}.

By Lemma~\ref{lem_equivalence} the above rewritings prove that any $\xor$-BACC 
is isomorphic to a $\xor$-BACC of same structure with at 
most two negative paths on its first two cycles.
Finally the equivalences \emph{(3)} and \emph{(4)} reduce the four base cases 
($++$,$+-$,$-+$,$--$) obtained this way to two: the positive case ($++$) and 
the negative case ($-+$).
\end{proof}

\begin{proof}(Point 2.)
 We now consider the interaction graph of a negative $\xor$-BACC of length 
$m$. By Equivalence \emph{(2)}, this network is isomorphic to a $\xor$-BACC of 
same structure with only one negative path on the first or on the second bottom 
half-cycle. Then, viewing the BACC upside-down, we can reuse the equivalences 
\emph{(6)} and \emph{(8)} alternatively so as to push this negative path to the 
right. Every time we apply the equivalences \emph{(6)} and \emph{(8)} 
successively the negative path is pushed 3 half-cycles to the right. Finally 
Equivalence \emph{(4)} tells us that if the negative path is pushed to the 
second last bottom half-cycle then the $\xor$-BACC is in the positive class. 
This can only happen if $m-1\equiv 1 \mod(3)$  or if $m-2\equiv 1 \mod(3)$, 
depending on if we start from the first or from the second bottom half-cycle 
respectively. In other words, this is the case if $m-1$ is not a multiple of 
$3$.

Moreover, the equivalences presented in Figure~\ref{eq_table} are 
exhaustive, \ie any other equivalences involving $\xor$-chains can be deduced 
from these eight equivalences. So, the argument above also proves that a 
positive $\xor$-BACC and a negative $\xor$-BACC cannot be isomorphic unless 
$m-1$ is a multiple of $3$. 
In other words, if $m-1\equiv 0\mod(3)$ there are always two isomorphism 
classes, the positive one and the negative one.
\end{proof}

 \begin{figure}
 \center
 \begin{tikzpicture}[->,>=latex,shorten >=1pt,node distance=1.5cm,minimum width = .8cm]
  \tikzstyle{vertex}=[draw = black,circle,text=black]
  \pgfmathsetmacro{\n}{6}
  \node[vertex] (1) {$\xor$};
  \path (1) edge [loop left] node[circle] {+} (1);

  \foreach \i in {2,...,\n}{
    \pgfmathtruncatemacro{\prec}{\i-1}
    \node[vertex] (\i) [right of = \prec] {$\xor$};
    \path   (\i) edge [bend left, near end] node[below,circle] {+} (\prec);
  }
 \path (\n) edge [loop right] node[circle] {--} (\n);

 \path (1) edge [bend left, near end] node[above, circle] {--} (2);
 \path (2) edge [bend left, near end] node[above, circle] {+} (3);
 \path (3) edge [bend left, near end] node[above, circle] {--} (4);
 \path (4) edge [bend left, near end] node[above, circle] {--} (5);
 \path (5) edge [bend left, near end] node[above, circle] {--} (6);
 \end{tikzpicture}
 
 $ \qquad \equiv \quad  (1)$
 
 \begin{tikzpicture}[->,>=latex,shorten >=1pt,auto,node distance=1.5cm,minimum width = .8cm]
  \tikzstyle{vertex}=[draw = black,circle,text=black]
  \pgfmathsetmacro{\n}{6}
  \node[vertex] (1) {$\xor$};
  \path (1) edge [loop left] node[circle] {+} (1);

  \foreach \i in {2,...,\n}{
    \pgfmathtruncatemacro{\prec}{\i-1}
    \node[vertex] (\i) [right of = \prec] {$\xor$};
    \path   (\i) edge [bend left, near end] node[circle] {+} (\prec);
  }
  
 \path (1) edge [bend left, near end] node[circle] {--} (2);
 \path (2) edge [bend left, near end] node[circle] {+} (3);
 \path (3) edge [bend left, near end] node[circle] {--} (4);
 \path (4) edge [bend left, near end] node[circle] {--} (5);
 
 \path (\n) edge [loop right] node[circle] {+} (\n);
 \path (5) edge [bend left, near end] node[circle] {+} (6);
 \end{tikzpicture}
  
 $ \qquad \equiv  \quad (8)$
 
 \begin{tikzpicture}[->,>=latex,shorten >=1pt,auto,node distance=1.5cm,minimum width = .8cm]
  \tikzstyle{vertex}=[draw = black,circle,text=black]
  \pgfmathsetmacro{\n}{6}
  \node[vertex] (1) {$\xor$};
  \path (1) edge [loop left] node[circle] {+} (1);

  \foreach \i in {2,...,\n}{
    \pgfmathtruncatemacro{\prec}{\i-1}
    \node[vertex] (\i) [right of = \prec] {$\xor$};
    \path   (\i) edge [bend left, near end] node[circle] {+} (\prec);
  }
 \path (\n) edge [loop right] node[circle] {+} (\n);

 \path (1) edge [bend left, near end] node[circle] {--} (2);
 \path (5) edge [bend left, near end] node[circle] {+} (6);

 \path (2) edge [bend left, near end] node[circle] {--} (3);
 \path (3) edge [bend left, near end] node[circle] {+} (4);
 \path (4) edge [bend left, near end] node[circle] {+} (5);

 \end{tikzpicture}
 
 $ \qquad  \equiv \quad (8)$
 
 \begin{tikzpicture}[->,>=latex,shorten >=1pt,auto,node distance=1.5cm]
  \tikzstyle{vertex}=[draw = black,circle,text=black]
  \pgfmathsetmacro{\n}{6}
  \node[vertex] (1) {$\xor$};
  \path (1) edge [loop left] node {$-$} (1);

  \foreach \i in {2,...,\n}{
  \pgfmathtruncatemacro{\prec}{\i-1}
  \node[vertex] (\i) [right of = \prec] {$\xor$};
  \path (\prec) edge [bend left, near end] node {$+$} (\i)
        (\i) edge [bend left, near end] node {$+$} (\prec);
}
 \path (\n) edge [loop right] node {+} (\n);
 \end{tikzpicture}
 \caption{Example of the rewrite of a cycle chain of type (1,2,2,2,2,2,1) into a negative cycle chain of type 
(1,2,2,2,2,2,1)}
 \label{fig_rewrite_ex}
 \end{figure}
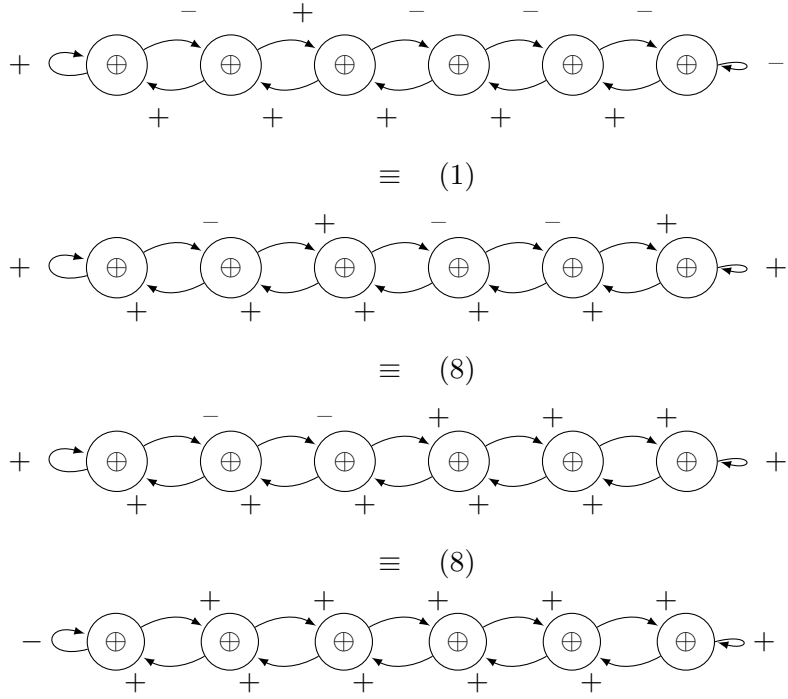

\subsection*{\textbf{ATG}}

For every type of $\xor$-BACCs, we now study the number of fixed points of each 
of their behavioural isomorphism classes so as to precise the general picture of 
their ATG given by Theorem~\ref{thm_reachability}.

\begin{prop}\label{lem_bacc_fp}
 A positive $\xor$-BACC of length $m$ and size $n$ has a unique fixed point, 
$0^n$, if $(m-1) \not\equiv 0 \mod (3)$ and has two fixed points, $0^n$ and 
$(101)^{\frac{m-1}{3}}$, if $(m-1)\equiv 0 \mod (3)$.
 A negative $\xor$-BACC (of length $m\equiv 1 \mod (3)$) has no fixed point.
\end{prop}

\begin{proof}
 In a stable configuration all the nodes of a given nude path have the same 
state, hence from now on we focus on determining the states of the intersection 
automata $o_k$.
 As this is done in Section~\ref{section_ex_BAF} for $\xor$-BAF, we determined 
the fixed points of a 
positive $\xor$-BACC by fixing the state of one of its automata and propagating 
the information induced until having to make a new choice or reaching a fixed 
point or a contradiction.
Here, we start by fixing Automaton $o_1$ (\ie the ``left most'' automaton) and 
by induction on the two possible cases ($x_{o_1} = 0$ and $x_{o_1} = 1$) we 
show that this completely determines the state of the other automata if $x$ is 
a fixed point.
\begin{enumerate}
 \item  if $x_{o_1} = 0$, then  $o_1$ is stable if and only if $x_{o_2} = 0$ 
and, recursively, for all $1<k\leq m-2$ , if 
$x_{o_{k-1}}=0$ and $x_{o_k} = 0$ then $o_k$ is stable if and only if 
$x_{o_{k+1}}=0$. Hence $0^m$ is the unique fixed point such that $x_0 = 0$.
 \item Similarly, if $x_{o_1} = 1$ then $o_1$ is stable if and only if $x_{o_2} 
= 0$. Then, we have three induction cases for all $1<k\leq m-2$:
 \emph{(1)} if $x_{o_{k-1}}=1$ and $x_{o_k} = 0$ then $o_k$ is stable if and 
only if $x_{o_{k+1}}=1$ ;
 \emph{(2)} if $x_{o_{k-1}}=0$ and $x_{o_k} = 1$ then $o_k$ is stable if and 
only if  $x_{o_{k+1}}=1$; 
 \emph{(3)} if $x_{o_{k-1}}=1$ and $x_{o_k} = 1$ then $o_k$ is stable if and 
only if  $x_{o_{k+1}}=0$. Hence 
the only way for the last intersection automaton, $o_{m-1}$, to be stable when 
$x_{o_1} = 1$ is that $(m-1)\equiv 0 (\mod 3)$, 
and the corresponding configuration is $(101)^{(m-1)/3}$. 
\end{enumerate}
This concludes the proof of the first statement. 

To show the second statement one only needs to realise that having a stable 
configuration for a negative $\xor$-BACC of length $m\equiv 1 \mod(3)$ amounts 
to having a stable configuration starting with a $1$ for a $\xor$-BACC of size 
$m-1$, which is impossible from the proof above.
 Indeed, if $x_{o_1} = 0$ then Automaton $o_1$ cannot be stable no matter 
the state value of Automaton $o_2$ in the configuration.
 Hence, if $x$ is a stable configuration $x_{o_1}$ must be $1$. This forces 
$x_{o_2}$ to be $1$ too (otherwise Automaton $o_1$ is not stable). 
 So, if $x$ is stable then $x_{o_2}\ldots x_{o_m}$ is a stable configuration 
starting with a $1$ for a positive $\xor$-BACC of size $m-1$. This is a 
contradiction. So there are no stable configurations for the negative 
$\xor$-BACC of length $m\equiv 1 \mod(3)$.
\end{proof}

According to Proposition~\ref{lem_bacc_classes}, if $\N$ is a $\xor$-BACC of 
length $m$ and size $n$ such that $m-1 \neq 0 \mod(3)$, then there is only one 
behavioural isomorphism class and so, similarly to what 
we have done for $\xor$-BAFs, it is possible to characterise completely the ATG 
of $\N$ using Proposition~\ref{lem_bacc_fp}: $G_\N^A$ has exactly one 
unreachable configuration, one fixed point, and one SCC of $2^{n}-2$ transient 
configurations.

The case where $m-1$ is a multiple of 3 is more complex because there are no 
easy ways to tell whether a network belongs to the positive or the negative 
class of its type, other than to compute its reduction graph as this is done in 
the proof of Proposition~\ref{lem_bacc_classes}. Moreover, the class of the 
reverse network also depends on the length of each half-cycle in the 
$\xor$-BACC, so describing each possible case would be tedious.
However, summarising the results above, we can still state that there is at 
most two fixed points and two unreachable configurations in the transition graph 
of a $\xor$-BACC of length $m-1 \equiv 0 \mod(3)$, or, to be more precise we 
can say that this transition graph has one of these four forms:
\begin{itemize}
 \item a SCC of size $2^n-4$, two fixed points and two unstable configurations (case $\N$ and  
$\N^R$ are from the positive class);
 \item a SCC of size $2^n-2$ and two fixed points  (case $\N$ is positive and $\N^R$ is negative);
 \item a SCC of size $2^n-2$ and two unreachable configurations  (case $\N$ is negative and $\N^R$ is positive);
 \item a SCC of size $2^n$ (case both $\N$ and $\N^R$ are negative).
\end{itemize}



\section{Interpretations and perspectives}

%
%

Through general results and their application to particular classes of interaction graphs, the present work launches 
the description of asymptotic dynamical behaviours of $\xor$-BANs under the asynchronous update mode.
By this means, it contributes to improve our understanding of the wild domain of 
non-monotonic Boolean automata networks. 
Theorem~\ref{thm_reachability} and Section~\ref{section_ex} suggest for 
example that non local monotony brings 
both entropy and stability to BANs since the high expressiveness of the resulting networks helps them to converge to 
fix points instead of getting stuck into larger attractors.
In the context of cellular reprogramming, the small number of attractors 
in $\xor$-BANs as well as the small number of irreversible configurations 
suggest that the genes involved in a $\xor$-cluster won't be good 
candidates for being reprogramming determinants~\cite{cpjs13}. Hence this might 
help to reduce the number of genes to consider.

The notion of behavioural isomorphism also reveals to be a powerful tool for factorising proofs when it comes to the 
study of a particular family of BANs. Even if finding a proper set of interaction graph rewritings may be a bit 
challenging, it results in a very interesting and comprehensive tool that highlights which characteristics of the 
interaction graphs really matter in the dynamical behaviours of the BANs.

 We believe that most of the results obtained could be refined or extended 
to some other types of ($\xor$)-BANs. For example it should be possible to 
allow some arcs between or inside the cycles of a $\xor$-BADC without changing 
the general shape of its corresponding ATG. These kinds of refinements draw a 
logical line for further works. 

Another interesting question would be directed to the study and comparison of asymptotic 
behaviours under different update modes. From this perspective, 
the algorithms we describe and the ATG we get for strongly connected 
$\xor$-BANs with an induced BADC of size greater than $3$ suggest that the addition of $k$-synchronism, that is when 
one allows $k$ automata to update simultaneously, make the set of unreachable configuration disappear if $k$ is 
greater than the size of the smallest cycle in an induced BADC of the network.


\bibliographystyle{plain}
\bibliography{flora_xorban_v0.bib}

\erase{
\section{Appendix}


\begin{tikzpicture}
 \node[circle, minimum width = 100pt] (o) {};
 \draw[] (o.center) -- (o.north) node[midway] {...};
 \draw[] (o.center) -- (o.south) node[midway]  {...};
 \draw[bend right,->] (o.center) to (o.60) to (o.center);
 \draw[bend right,->] (o.center) to (o.120) to (o.center);
 \draw[bend right,->] (o.center) to (o.-60) to (o.center);
 \draw[bend right,->] (o.center) to (o.-120) to (o.center);
 \draw[->] (o.180) to (o.center) ;
 \draw[->] (o.center) to (o.0);
 
\end{tikzpicture}
}

\end{document}